\newif\ifDRAFT 
\DRAFTfalse

\documentclass[11pt]{article}
\usepackage[margin=1in]{geometry}
\usepackage{graphicx} 
\usepackage{mdframed}
\usepackage{amsfonts, amsmath, amssymb, amsthm}
\usepackage{mathrsfs}  
\usepackage{algorithm}
\usepackage[noend]{algpseudocode}
\usepackage{hyperref}
\usepackage{color}
\usepackage{array}
\usepackage{cleveref}
\usepackage[inkscapeformat=pdf]{svg}
\usepackage{multirow}
\usepackage{mathtools}

\usepackage{tabulary} 
\usepackage{booktabs}
\usepackage{arydshln}

\ifDRAFT
    \usepackage{lineno}
    \linenumbers
\fi

\theoremstyle{plain}
\newtheorem{theorem}{Theorem}[section]
\newtheorem{lemma}[theorem]{Lemma}

\newtheorem{observation}[theorem]{Observation}

\newtheorem{conjecture}{Conjecture}[section]

\theoremstyle{definition}
\newtheorem{definition}[theorem]{Definition}

\newtheorem{remark}[theorem]{Remark}

\crefname{equation}{Eqn.}{Eqns.}

\newcommand{\paren}[1]{\left( #1 \right)}

\newcommand{\sparen}[1]{\left[ #1 \right]}

\DeclareSymbolFont{bbold}{U}{bbold}{m}{n}
\DeclareSymbolFontAlphabet{\mathbbold}{bbold}

\newcommand{\C}{\mathcal{C}}
\newcommand{\E}{\mathbb{E}}
\newcommand{\F}{\mathcal{F}}

\newcommand{\poly}{\operatorname{poly}}

\newcommand{\dist}{\operatorname{dist}}

\newcommand{\conf}{\mathsf{conf}}

\usepackage[normalem]{ulem} 

\title{Color Fault-Tolerant Spanners}

\ifDRAFT
    \author{Anonymous Authors}
\else
    \author{
     Asaf Petruschka\thanks{Supported by the European Research Council (ERC) under the European Union’s Horizon 2020 research and innovation programme, grant agreement No. 949083, and by the Israeli Science Foundation (ISF), grant 2084/18.}
     \qquad
     Shay Sapir\thanks{This research was partially supported by the Israeli Council for Higher Education (CHE) via the Weizmann Data Science Research Center.}
     \qquad
     Elad Tzalik\\
     Weizmann Institute of Science
     \\ \texttt{\{asaf.petruschka,shay.sapir,elad.tzalik\}@weizmann.ac.il}}
\fi
\date{}

\begin{document}

\maketitle

\pagenumbering{gobble}

\begin{abstract}

We initiate the study of spanners in arbitrarily vertex- or edge-colored graphs (with no ``legality'' restrictions), that are resilient to failures of \emph{entire color classes}.
When a color fails, all vertices/edges of that color crash.
An $f$-color fault-tolerant ($f$-CFT) $t$-spanner of an $n$-vertex colored graph $G$ is a subgraph $H$ 
that preserves distances up to factor $t$, even in the presence of at most $f$ color faults. 
This notion generalizes the well-studied $f$-vertex/edge fault-tolerant ($f$-V/EFT) spanners.
The size of an $f$-V/EFT spanner crucially depends on the number  $f$ of vertex/edge faults to be tolerated.
In the colored variants, even a single color fault can correspond to an unbounded number of vertex/edge faults.

The key conceptual contribution of this work is in showing that the size (number of edges) required by an $f$-CFT spanner is in fact comparable to its uncolored counterpart, with no dependency on the size of color classes.
We provide optimal bounds on the size required by $f$-CFT $(2k-1)$-spanners, as follows:
\begin{itemize}
    \item When vertices have colors, we show an upper bound of $O(f^{1-1/k} n^{1+1/k})$ edges.
    This precisely matches the (tight) bounds for $(2k-1)$-spanners resilient to $f$ \emph{individual} vertex faults [Bodwin et al., SODA 2018; Bodwin and Patel, PODC 2019].

    \item For colored edges, we show that $O(f n^{1+1/k})$ edges are always sufficient.
    Further, we prove this is tight, i.e., we
    provide an $\Omega(f n^{1+1/k})$ (worst-case) lower bound.
    The state-of-the-art bounds known for the corresponding uncolored setting of edge faults are (roughly) $\Theta(f^{1/2} n^{1+1/k})$ [Bodwin et al., SODA 2018; Bodwin, Dinitz and Robelle, SODA 2022].

    \item We also consider a \emph{mixed} model where both vertices and edges are colored.
    In this case, we show tight $\Theta(f^{2-1/k} n^{1+1/k})$ bounds.
\end{itemize}
Thus, CFT spanners exhibit an interesting phenomenon: while (individual) edge faults are ``easier" than vertex faults, edge-color faults are ``harder" than vertex-color faults.

Our upper bounds are based on a generalization of the blocking set technique of [Bodwin and Patel, PODC 2019] for analyzing the (exponential-time) greedy algorithm for FT spanners.
We complement them by providing \emph{efficient} constructions of CFT spanners with similar size guarantees, based on the algorithm of [Dinitz and Robelle, PODC 2020].

\end{abstract}


\newpage
\pagenumbering{arabic}

\section{Introduction}

Let $G$ be an $n$-vertex weighted graph.
A \emph{spanner} of $G$ is a subgraph that approximately preserves distances.
Formally, a subgraph $H$ is a $t$-spanner of $G$ if $\dist_H (u,v) \leq t \cdot \dist_G (u,v)$ for all $u,v\in V$ (where $\dist_X$ denotes the shortest path distance in a graph $X$).
Since their introduction by Peleg and Ullman \cite{PelegU88} and Peleg and Sch\"{a}ffer \cite{PelegS89}, spanners have been found to be extremely useful for a wide variety of applications, including network tasks like routing and synchronization \cite{PelegU89a,ThorupZ01,Cowen01,CowenW04,PelegU88}, distance estimation \cite{ThorupZ05,BaswanaS04}, preconditioning linear systems \cite{ElkinEST08}, and many others.

In real-world scenarios, spanners are frequently employed in systems whose components are susceptible to occasional breakdowns.
It is therefore desirable to have spanners possessing resilience to such failures, leading to the notion of \emph{fault-tolerant (FT) spanners}.
Two extensively studied types of faults are \emph{edge faults} and \emph{vertex faults}.
Spanners in the edge/vertex fault-tolerant (E/VFT) settings are defined as follows:

\begin{definition}[E/VFT Spanner]
    An $f$-EFT (VFT) $t$-spanner of $G$ is a subgraph $H$ such that for every set $F$ of at most $f$ edges (vertices) in $G$, it holds that $H-F$ is a $t$-spanner of $G-F$.
\end{definition}

E/VFT spanners have received major interest in recent years; see, e.g.,
\cite{ChechikLPR10,DinitzK11,BodwinDPW18,BodwinP19,DinitzR20,BodwinDR21,BodwinDR22,Parter22} and references therein.

In this work, we introduce and study \emph{color fault tolerant} models.
The framework is the following: the edges or vertices are arbitrarily partitioned into classes, or equivalently, associated with colors, and $f$ such \emph{color classes} might fail.
Each of the E/VFT settings gets a colored counterpart setting, denoted E/V\underline{\textbf{C}}FT (the C stands for `color').
The standard E/VFT settings are obtained as special cases, by giving each edge/vertex a different color.

Faulty color classes have been used to model 
Shared Risk Resource Groups (SRRG), which arise in many practical contexts such as optical telecommunication networks and multi-layered networks; see~\cite{CoudertDPRV07,Kuipers12,ZPT11} and  references therein for a detailed discussion.
Previous work in this context mainly focused on colored variants of classical graph problems (and their hardness).
A notable such problem is \emph{diverse routing}, where the goal is to determine if two (or more) color disjoint paths between two vertices \cite{Hu03,EGBERRL03,MENTA02} exist (and if so,  finding them).
In such algorithmic contexts, replacing a dense colored graph by a sparse CFT spanner could potentially help to improve algorithmic performance, similarly to V/EFT spanners for uncolored graphs.

\paragraph{Sparse Connectivity Certificates.}
Connectivity certificates are subgraph notions that are closely related to FT spanners.
Intuitively, a \emph{$\lambda$-edge (or $\lambda$-vertex) connectivity certificate} serves as concise ``proof'' to establish the $\lambda$-edge/vertex connectivity of a graph $G$.
Formally, it is a subgraph $H$ of $G$ such that 
every two vertices $u,v$ are $\lambda$-edge/vertex connected in $G$ if and only if they are $\lambda$-edge/vertex connected in $H$.
The concept of connectivity certificates was presented by Nagamochi and Ibaraki \cite{NagamochiI92}, who showed the existence of
$\lambda$-connectivity certificates with at most $\lambda(n-1)$ edges for every $n$-vertex graph.
These play a crucial role in various graph algorithms 
\cite{Matula93,CheriyanKT93,KargerM97,GhaffariK13,ForsterNYSY20,LiNPSY21}.
Every $\lambda$-V/EFT $(2k-1)$-spanner, for every $1 \leq k < \infty$, is \emph{by definition} also a $\lambda$ connectivity certificate, as observed e.g.\ in \cite{Par19,Parter22}.
Such a spanner with $O(\lambda^{1-1/k} n^{1+1/k})$ edges always exists \cite{BodwinP19}, so choosing $k = \Theta(\log n)$ yields sparse $\lambda$-connectivity certificates of $O(\lambda n)$ size.  

Due to the min-cut max-flow theorem, edge/vertex-connectivity can be defined equivalently using cuts and flows.
In the colored case, this relationship no longer holds: the minimum set of colors whose removal disconnects two vertices can be arbitrarily larger than the maximum number of color-disjoint paths connecting them \cite{CoudertDPRV07}. 
Hence, ``color-connectivity certificates'' may take two different meanings for cuts and flows.
CFT Spanners can then be seen as color-connectivity certificates \emph{in the cut sense}.
To avoid confusion, we call the latter $\lambda$-CFT-connectivity certificates.
Similarly to their uncolored counterparts, such CFT-connectivity certificates may provide useful tools when dealing with colored graphs.
As this work focuses on the stronger notion of CFT spanners, (wort-case) optimal results on the size of CFT-connectivity certificates are implied as direct corollaries.

\paragraph{A Mixed Setting.}
We additionally consider another possible setting, which we call \emph{mixed}, that allows for both edge faults and vertex faults.
That is, in the uncolored setting, the faulty-set $F$ is a subset of $V(G) \cup E(G)$ of size $\leq f$.
In the colored setting, both edges and vertices have colors (and a color may appear on both vertices and edges).
We use the shorthand MFT for the uncolored mixed setting, and MCFT for its colored counterpart.
To the best of our knowledge, MFT spanners have not been explicitly defined,
but many known results for the E/VFT settings implicitly apply to them.

\subsection{Our Contribution}

We initiate the study of FT spanners supporting $f$ color faults in colored graphs.
A na\"ive approach for handling color faults is to treat them as many edge/vertex faults. 
If each color class contains at most $\Delta$ edges/vertices, then a $\Delta f$-FT spanner for edge/vertex faults is clearly  an $f$-\underline{\textbf{C}}FT spanner.
While this solution might be plausible for small $\Delta$, it is useless for large $\Delta$.
As the coloring is arbitrary, it may even be that $\Delta = \Omega(n)$.
Thus, the resulting spanner from this na\"ive approach might have $\Omega(n^2)$ edges even when $f=1$, i.e., when only a single color fault may occur.
Conceptually, the main contribution of our work is showing that in the realm of FT spanners, one can handle $f$ color faults almost as efficiently as $f$ \emph{individual} edge/vertex faults, without any dependence on $\Delta$.

Our main results regarding CFT spanners are:

\begin{theorem}\label{thm:main}
    For every $k,f\geq 1$, and every colored $n$-vertex graph $G$, there exists an $f$-CFT $(2k-1)$-spanner with the following size bounds:
    \begin{enumerate}
        \item $O(fn^{1+1/k})$ edges in the edge-colored setting, 
        \item $O(f^{1-1/k} n^{1+1/k})$ edges in the vertex-colored setting, and
        \item $O(f^{2-1/k} n^{1+1/k})$ edges in the mixed setting of both edge and vertex colors.
    \end{enumerate}
    Moreover, the above bounds are tight assuming the Erd\H{o}s Girth Conjecture.
\end{theorem}

\Cref{table:size-bounds} presents the known bounds for non-colored FT-spanners, compared to the color settings.
The bounds of \Cref{thm:main} are based on the \emph{exponential-time} FT greedy algorithm. 
We further provide polynomial-time constructions of CFT spanners with nearly optimal size guarantees (up to $\poly(k)$ factors).
We obtain $\lambda$-CFT connectivity certificates as immediate corollaries:
\begin{theorem}
    Every colored $n$-vertex graph $G$ has a subgraph $H$ such that, for every two vertices $u,v$, and every set $C$ of at most $\lambda$ colors, $u,v$ are connected in $H-C$ iff they are connected in $G-C$, where
    $|E(H)| = O(\lambda n)$ if the graph has either colored edges or colored vertices (but not both), and
    $|E(H)| = O(\lambda^2 n)$ when both vertices and edges are colored.
    This is (worst-case) optimal.
\end{theorem}

\begin{table}
\centering
\scalebox{1.05}{
\begin{tabular}{@{}ccccc@{}}
\toprule
\textbf{Setting} & \textbf{Upper bound}                                                                                                                  & \textbf{Ref.}        & \textbf{Lower bound}                             & \textbf{Ref.}          \\ \midrule\midrule
EFT     & \begin{tabular}[c]{@{}l@{}}$O_k (f^{\tfrac{1}{2}-\tfrac{1}{2k}} \cdot n^{1+\tfrac{1}{k}} + fn)$, odd $k$\\ $O_k (f^{\tfrac{1}{2}} \cdot n^{1+\tfrac{1}{k}} + fn)$, even $k$\end{tabular} & \cite{BodwinDR22} & $\Omega(f^{\tfrac{1}{2}-\tfrac{1}{2k}} \cdot n^{1+\tfrac{1}{k}} + fn)$ & \cite{BodwinDPW18} \\ 
ECFT    & $O(f\cdot n^{1+\tfrac{1}{k}})$   & \textbf{New} & $\Omega(f \cdot n^{1+\tfrac{1}{k}})$         & \textbf{New}       \\ \midrule
VFT     & \multirow{2}{*}{$O(f^{1-\tfrac{1}{k}} \cdot n^{1+\tfrac{1}{k}})$}                                                                                                      & \cite{BodwinP19}  & \multirow{2}{*}{$\Omega(f^{1-\tfrac{1}{k}}\cdot n^{1+\tfrac{1}{k}})$}            & \multirow{2}{*}{\cite{BodwinDPW18}} \\ 
VCFT    &    & \textbf{New} &              &  \\ \midrule
MFT     & $O(f^{1-\tfrac{1}{k}} \cdot n^{1+\tfrac{1}{k}})$                                                                                                      & \cite{BodwinP19}  & $\Omega(f^{1-\tfrac{1}{k}} \cdot n^{1+\tfrac{1}{k}})$            & \cite{BodwinDPW18} \\
MCFT    & $O(f^{2-\tfrac{1}{k}} \cdot n^{1+\tfrac{1}{k}})$ & \textbf{New} & $\Omega(f^{2-\tfrac{1}{k}} \cdot n^{1+\tfrac{1}{k}})$ & \textbf{New}       \\ \bottomrule
\end{tabular}
}
\caption{Bounds on the worst-case size required by $f$-FT $(2k-1)$-spanners.
All lower bounds are conditioned on the Erd\H{o}s Girth Conjecture \cite{Erdos63}.
The lower bounds in the ECFT and MCFT settings cannot exceed $\Omega(n^2)$ for simple graphs, but we omit the $\min\{\cdot, n^2\}$ expression for ease of reading.
The bounds in the MFT settings are implicit by the cited works.
}\label{table:size-bounds}
\end{table}

\paragraph{On FT Spanners for Multi-Graphs.}
Interestingly, the bounds we provide for \emph{simple} colored graphs match the bounds required by E/V/MFT spanners for uncolored \emph{multi-graphs}.
All of our results apply both to simple and multi-graphs. 
This is in contrast to the uncolored setting, where ``multi'' FT-spanners are denser than ``simple'' ones.
See discussion in \Cref{sect:multi-graphs}.

\paragraph{Color Lists.}
Consider the following generalization of the color-faults model: Instead of having a single color per vertex/edge, each of them now has a \emph{list} of colors to which it is sensitive.
That is, the failure of any single color from the list makes the vertex/edge crash.
Notice that each of the E/V/MCFT setting is obtained as a special case by restricting the sizes of the lists; e.g., the ECFT setting is obtained by restricting edges to have lists of length 1, and the vertices to have empty 0-length lists.
In fact, this generalized setting is also considered in the context of Shared Risk Resource Groups (SRRG); see \cite{CoudertDPRV07} and references therein. 
By studying this variant, we obtain the following generalization of \Cref{thm:main}:

\begin{theorem}\label{thm:color-lists}
    Fix constant integers $\mu, \nu \geq 0$.
    Let $G$ be an $n$-vertex graph, where each edge (resp., vertex) has a list of at most $\mu$ (resp., $\nu$) colors, such that the failure of \emph{any} color in the list makes the edge (resp., vertex) fail.
    Then, there exist an $f$-CFT $(2k-1)$-spanner with $O(f^{\mu + \nu(1-1/k)} n^{1+1/k})$ edges.
    Moreover, this is tight assuming the Erd\H{o}s Girth Conjecture.
\end{theorem}

The bounds of \Cref{thm:main} are obtained by taking $\mu, \nu \in \{0,1\}$.
The proof is a relatively straightforward extension of the proofs for the E/V/MCFT settings, using inductive arguments;
it is given in \Cref{sect:color-lists}.

\subsection{Technical Overview}\label{sect:technical_overview}

\paragraph{Upper bounds (\Cref{sect:upper-bounds}).}
To prove size upper bounds
for CFT spanners, we build upon the powerful yet strikingly simple \emph{blocking set} technique of Bodwin and Patel~\cite{BodwinP19}.
This is a method for analyzing the spanner $H$ given by the (exponential-time) \emph{FT greedy algorithm}, used in \cite{BodwinP19} to obtain optimal bounds in the VFT setting.
For clarity of exposition, we focus our discussion here only on the corresponding vertex-colored setting.

Bodwin and Patel define a blocking set as a collection of edge-vertex pairs, such that each short ($\leq 2k$) cycle in $H$ contains one of the pairs.
On an intuitive level, a graph that admits a small blocking set is close to having high girth, forcing it to be sparse.
The FT greedy algorithm naturally induces a blocking set of size $f |E(H)|$, which ultimately yields the upper bound on the size of $H$ (as explained next).
The trivial adaption of this algorithm to the CFT settings then leads to the ``right'' extension of a blocking set: a collection of edge-\textbf{color} pairs such that for every short cycle $\Sigma$, at least one pair has both its edge and its color appearing on $\Sigma$. 

The analysis of $|E(H)|$ is by considering a random subgraph of $H$ induced on a fraction of $\approx 1/f$ vertices, and removing the edges of surviving blocking-set-pairs.
The resulting subgraph, denoted $\hat{H}_S$, has $O(n/f)$ vertices and large girth, and therefore $O((n/f)^{1+1/k})$ edges according to the folkore \emph{Moore Bounds}.
On the other hand, direct calculations show that $\Omega( |E(H)|/f^2 )$ edges are expected to survive, hence $|E(H)| = O(f^{1-1/k}n^{1+1/k})$.
Again, there seems to be a natural adaptation for (vertex) colors: sample a random fraction of $\approx 1/f$ \emph{color classes}.
The problem that arises is that even though $\hat{H}_S$ has $O(n/f)$ vertices \emph{in expectation}, it is quite probable to have $\Omega(n)$ vertices (say, if a color class of size $\Omega(n)$ is sampled).
More conceptually, the problem is that vertices cannot be sampled independently, preventing us from using concentration bounds.
Na\"ively carrying out the same analysis with a trivial $O(n)$ bound on the number of sampled vertices yields only an $O(f^2 n^{1+1/k})$ bound on the size of the spanner $H$.

We show that even in the independence-deprived circumstances of the VCFT setting, the same bounds of the VFT setting can be obtained.
Our analysis is carefully designed to harness the little independence we do have, between different color classes.
It has three main ingredients.
\begin{itemize}
    \item[(1)] We first show that effectively, we may assume that greedy spanner $H$ is \emph{color-bipartite}: a bipartite graph where each color class is fully contained in one of the sides.

    \item[(2)] To complement this, we use a ``symmetric'' version of the Moore Bounds: a bipartite graph with sides $L,R$ and girth $> 2k$ has $O(|L||R|^{1/k}+|R||L|^{1/k})$ edges.%
    \footnote{Stronger bipartite Moore Bounds exist~\cite{Hoory02}, but this weaker bound suffices for our needs.}

    \item[(3)] Combining the independence between the color classes \emph{from different sides} with the Moore Bounds, we show that the random graph $\hat{H}_S$ has $O((n/f)^{1+1/k})$ edges in expectation.
\end{itemize}

\paragraph{Lower bounds (\Cref{sect:lower-bounds}).}
Our lower bounds, as most spanner lower bounds, are conditioned on the famous and widely believed Erd\H{o}s Girth Conjecture \cite{Erdos63},
implying the existence of an $n$-vertex graph $G_{ER}$ with $\Omega(n^{1+1/k})$ edges that has no proper subgraph as a $(2k-1)$-spanner.

For edge-colored graphs, our strategy is to pack $f$ copies of an $n$-vertex Erd\H{o}s graph $G_{ER}$, denoted $G_1, \dots, G_f$, on the same set of vertices, with little overlap.
We then consider the union graph $\bigcup_{i} G_i$ (ignoring overlaps), where each edge originating from $G_i$ is given the color $i$.
This graph has $\Omega(fn^{1+1/k})$ edges, and no proper $f$-ECFT $(2k-1)$-spanner.

In the mixed-colored setting, we ``encode'' $f^2$ disjoint $f$-ECFT lower bound instances in a graph with $fn$ properly colored vertices, such that each instance can be isolated by failing $2f$ vertex-colors. Thus, this graph has no proper $3f$-MCFT spanner.

\paragraph{Efficient algorithms (\Cref{sect:efficient-constructions}).}
We show efficient CFT spanner constructions based on the technique of Dinitz and Robelle \cite{DinitzR20}.
Their approach was to modify the greedy algorithm in the E/VFT settings to have polynomial running time, while ensuring it still induces a small blocking set, only larger by an $O(k)$ factor than that of the ``vanilla'' FT greedy algorithm.
Thus, the size analysis based on blocking sets still applies.
Our (existential) upper bounds generalize this last component 
to colored settings.
We then augment the \cite{DinitzR20} algorithm to produce small \emph{colored} blocking sets, and thus achieve efficient constructions.

\paragraph{Remark on Optimality Beyond the Girth Conjecture.}
Let $\gamma(n,2k)$ denote the maximum number of edges in an $n$-vertex graph of girth $> 2k$.
The Moore Bounds show that  $\gamma(n,2k) = O(n^{1+1/k})$, and Erd\H{o}s' Girth Conjecture is that this is tight up to a constant factor.
Our lower bounds can be formulated purely in terms of the $\gamma$ function: $\Omega(f \cdot \gamma(n,2k))$ edges in the ECFT setting, $\Omega(f^2 \cdot \gamma(n/f, 2k))$ edges in VCFT setting, and $\Omega(f^3 \cdot \gamma(n/f, 2k))$ edges in MCFT setting.
As for the upper bounds:
In the ECFT case, the upper bound matches the lower bound also in terms of $\gamma$.
In the VCFT and MCFT settings, the upper bounds are matching assuming only much weaker conditions on $\gamma$ than the Girth Conjecture; it is enough to assume that $\gamma(n, 2k) = \Theta(n \cdot \beta(n,2k))$, where $\beta(\cdot, 2k)$ is \emph{concave} (the Girth Conjecture is that $\beta(n,2k) = n^{1/k}$).

\section{Preliminaries: A Unified FT Framework}\label{sect:prelim}
We now present a unified formal framework for all six E/V/M(C)FT settings, to help avoid tedious repetitions later on.

Henceforth, assume every edge/vertex $x \in E(G) \cup V(G)$ is associated with a color $c(x) \in \C$, where $\C$ is some finite set called the \emph{color palette}. 
The coloring is arbitrary, and there are no ``legality'' restrictions.
For example, two neighboring vertices may share the same color.
It is also allowed that the same color appears both on edges and on vertices.
For $c \in \C$, we denote the sets of $c$-colored edges and vertices by $E_c(G) = \{e \in E(G) \mid c(e) = c\}$ and $V_c (G) = \{v \in V(G) \mid c(v) = c\}$.
The edges of $G$ are weighted, where $w(e) = w(u,v)$ denotes the weight of edge $e = \{u,v\}$.

The \emph{fault universe} $\F$ is defined separately for each of the fault-tolerant settings:
\begin{table}[H]
\centering
\begin{tabular}{@{}ccccccc@{}}
\toprule
\textbf{Setting} & EFT    & VFT    & MFT              & ECFT                   & VCFT                     & MCFT                             \\ \midrule
$\F = $    & $E(G)$ & $V(G)$ & $E(G) \cup V(G)$ & $\{E_c(G)\}_{ c\in \C}$ & $ \{V_c(G)\}_{ c\in \C}$ & $\{E_c(G) \cup V_c(G)\}_{ c\in \C}$ \\ \bottomrule
\end{tabular}
\end{table}
We say that an edge $e = \{u,v\} \in E(G)$ is \emph{damaged} by $x \in \F$ if the failure of $x$ causes $e$ to fail.
Formally, $x$ damages $e$ if: $x=e$ in EFT, $x \in \{u,v\}$ in VFT, $x \in \{e,u,v\}$ in MFT, $e \in x$ in ECFT, $\{u,v\} \cap x \neq \emptyset$ in VCFT, and $\{e,u,v\} \cap x \neq \emptyset$ in MCFT.
Note that in the colored settings, there is a one-to-one correspondence between the fault universe $\F$ and the set of colors $\C$. 
So, when $S \subseteq \F$, we say that $S$ contains the color $c$, and abuse notation by writing $c \in S$, when $S$ contains the element corresponding to $c$ from the fault universe $\F$ (which depends on the setting). 
We denote by $G-S$ the subgraph of $G$ obtained by deleting all edges damaged by some $x \in S$, and by $G[S] \coloneqq G - (\F - S)$, which is the subgraph obtained by keeping all edges that can only be damaged by faults from $S$.

We now formally state the general FT spanner definition, applying to all settings:

\begin{definition}[General FT Spanner]
    An $f$-FT $t$-spanner of $G$ with respect to a fault universe $\F$ (defined by the given setting) is a subgraph $H$ such that for every $F \subseteq \F$ with $|F|\leq f$, it holds that $H-F$ is a $t$-spanner of $G-F$.
\end{definition}

\section{Upper Bounds}\label{sect:upper-bounds}

The FT variant of the greedy spanner algorithm of Alth\''ofer et al.~\cite{AlthoferDDJS93}, presented as \Cref{alg:greedy-spanner}, has been central in the study of spanners in general, and FT spanners in particular.
\begin{algorithm}[H]
\caption{$\mathsf{FTGreedySpanner}(G,k,f)$}\label{alg:greedy-spanner}

\begin{algorithmic}[1]
    \State $H \gets (V(G), \emptyset)$
    \For{each $e = \{u,v\} \in E(G)$ in increasing order of weight}
        \If{there is $F \subseteq \F$, $|F| \leq f$, not damaging $e$, such that $\dist_{H-F} (u,v) > (2k-1)w(e)$}
            \State $H \gets H \cup \{e\}$
        \EndIf
    \EndFor
    \State \Return $H$
\end{algorithmic}
\end{algorithm}

Aside from its immediate correctness, the greedy approach has been proven to produce strong and often optimal upper bounds on the (worst-case) size of E/VFT spanners (e.g., see 
\cite{BodwinP19,BodwinDR22}),
making it an ``immediate suspect'' to investigate for getting upper bounds on the size of CFT spanners as well.
The major downside of the na\"ive FT greedy algorithm is its exponential-in-$f$ running time, which makes it intractable for many algorithmic applications.
However, it is often the case that modified polynomial-time greedy approaches (with suitably modified analysis) can be devised to produce very close size bounds \cite{BodwinDR21,DinitzR20}.

In this section, we provide existential size bounds on CFT spanners via the naive FT greedy algorithm, using the \emph{blocking set} technique of Bodwin and Patel \cite{BodwinP19}.
We show:

\begin{theorem}\label{thm:CFT-upper-bounds}
    The number of edges $m$ in the graph $H$ constructed by the FT greedy algorithm (\Cref{alg:greedy-spanner}) can be bounded as follows:
    \begin{enumerate}
        \item In the ECFT setting, $m = O(f n^{1+1/k})$.
        \item In the VCFT setting, $m = O(f^{1-1/k} n^{1+1/k})$.
        \item In the MCFT setting, $m = O(f^{2-1/k} n^{1+1/k})$
    \end{enumerate}
\end{theorem}

\subsection{Blocking Sets and Random Blocked Subgraphs}

Denote by $H$ the output of \Cref{alg:greedy-spanner}, and by $m$ its number of edges.
The size analysis of $H$ hinges on a generalized definition of blocking sets, capturing all six E/V/M(C)FT settings.

\begin{definition}[Blocking Set]\label{def:blocking-set}
    A $2k$-\emph{blocking set} of $H$ is a subset  $B \subseteq E(H) \times \F$ that sastifies the following conditions:
    \begin{enumerate}
        \item[(i)] If $(e,x) \in B$, then $e$ is not damaged by $x$.
        \item[(ii)] For every cycle $\Sigma$ in $H$ on $\leq 2k$ edges, there exist $(e,x) \in B $ such that $e \in \Sigma$, and some (other) edge of $\Sigma$ is damaged by $x$. Abusing notation, we denote this condition by $(e,x) \in B \cap \Sigma$.
    \end{enumerate}
\end{definition}

The same arguments as in \cite{BodwinP19} imply the existence of small blocking set for $H$:
\begin{lemma}[Modification of \protect{\cite[Lemma 3]{BodwinP19}}]\label{lem:blocking-set}
    There is a $2k$-blocking set $B$ of $H$ such that $|B| \leq f m$.
    Moreover, for every edge $e \in E(H)$, there are at most $f$ pairs of the form $(e,x)$ in $B$.
\end{lemma}
\begin{proof}
    For every $e=\{u,v\}\in E(H)$, let $F_e \subseteq \F$, $|F_e|\leq f$ be the set of faults not damaging $e$ that caused $e$'s insertion.
    I.e., at the time $e$ was considered by \Cref{alg:greedy-spanner}, it was the case that $\dist_{H-F_e}(u,v)>(2k-1)w(e)$.
    Define
    \[
    B=\{(e,x)\mid e\in E(H),x\in F_{e}\}.
    \]
    Every edge $e \in E(H)$ participates in at most $f$ pairs $(e,x)\in B$,
    and condition (i) of \Cref{def:blocking-set} holds.
    We show condition (ii).
    Let $\Sigma$ be a cycle on $\leq 2k$ edges in $H$.
    Consider the heaviest edge $e=\{u,v\}$ on $\Sigma$.%
    \footnote{
        Ties are broken as in \Cref{alg:greedy-spanner}. That is, $e$ is edge in $\Sigma$ that appeared last in the ordering by which \Cref{alg:greedy-spanner} went over the edges of $G$.
    }
    Then $\Sigma - e$ was already present in $H$ before $e$ was added, forming a $u$-$v$ path of length $\leq (2k-1)w(e)$.
    But, before adding $e$, we had $\dist_{H-F_e}(u,v)>(2k-1)w(e)$, so some $x\in F_e$ must damage an edge of $\Sigma-e$. Thus, $(e,x)\in B \cap \Sigma$.
\end{proof}

\begin{remark}\label{remark:blocking-set-proof}
    In fact, in the following proof of \Cref{thm:CFT-upper-bounds}, we only use that the output graph of the FT-greedy algorithm admits a small blocking set, as shown in \Cref{lem:blocking-set}.
    That is, our proof shows that if $H$ is some $n$-vertex graph that admits a $2k$-blocking set $B$ where each edge $e \in E(H)$ participates in at most $f$ pairs $(e,x) \in B$, then $|E(H)|$ can be bounded from above in the E/V/MCFT settings as stated in \Cref{thm:CFT-upper-bounds}.
\end{remark}

\begin{remark}\label{remark:blocking-set-for-subgraphs}
    Note that \Cref{lem:blocking-set} stays true also for subgraphs of $H$, by deleting the pairs from $B$ that correspond to the edges missing in the subgraph.
\end{remark}

To analyze the size of $H$, we consider the expected number of edges in a randomly-formed subgraph $\hat{H}_S$ of $H$, which we call a \emph{random blocked subgraph}.
Informally, we bound this expectation from below by $|E(H)|/\poly(f)$, and from above by a function of $n,f$ alone. Rearranging yields the final bounds.

\begin{definition}[$p$-Random Blocked Subgraph]\label{def:random-blocked-subgraph}
    Let $S \subseteq \F$ be a random subset obtained by sampling each $x \in \F$ into $S$ independently with probably $p$.
    Consider $H[S]$, the subgraph of $H$ obtained by removing all edges that are damaged by $\F-S$.
    Essentially, the \emph{$p$-random blocked subgraph} $\hat{H}_S$ is  obtained from $H[S]$ by removing the edges of surviving elements of the blocking set. Formally:
    \begin{align*}
        B_S &= \{(e,x)\in B \mid e \in E(H[S]), x \in S\}, \\
        E(B_S) &= \{ e \mid (e,x) \in B_S \text{ for some $x\in S$}\}, \\
        \hat{H}_S &= H[S] - E(B_S).
    \end{align*}
    The number of edges in $\hat{H}_S$ is denoted by $\hat{m}$.
\end{definition}

The point of this process is the following:

\begin{observation}\label{obs:girth}
    A $p$-random blocked subgraph $\hat{H}_S$ has girth $\geq 2k+1$ with probability $1$.
\end{observation}
\begin{proof}
    Suppose $\Sigma$ is a cycle of $\leq 2k$ edges in $H[S]$
    (and thus also in $H$). 
    As $B$ is a blocking set, there is some $(e,x) \in B \cap \Sigma$.
    Since $\Sigma$ is sampled into $H[S]$, then $e \in E(H[S])$ and $x \in S$, showing that $(e,x) \in B_S$.
    Thus, $e$ is removed in $\hat{H}_S$, so $\Sigma$ does not survive in $\hat{H}_S$.
\end{proof}

\subsection{The ECFT Setting}

We start by analyzing the ECFT setting, in which a trivial analysis yields optimal results. It serves as a good warm-up for the other settings.
The analysis uses the folklore \emph{Moore Bounds}:

\begin{lemma}[Moore Bounds]\label{lem:Moore-bounds}
    An $n$-vertex graph with girth $\geq 2k+1$ has $O(n^{1+1/k})$ edges.
\end{lemma}

Now,
let $\hat{H}_S$ be a $1/(2f)$-random blocked subgraph.
Consider an edge $e \in E(H)$ of the FT greedy spanner $H$.
Clearly, $e \in H[S]$ iff $c(e)\in S$, which happens with probability $1/(2f)$.
Next, consider a pair $(e,c) \in B$.
As $c$ does not damage $e$, meaning $c(e) \neq c$, then $(e,c)\in B_S$ iff both $c$ and $c(e)$ are sampled, which happens with probability $1/(2f)^2$.
Thus,
\begin{align*}
    \E[\hat{m}] &= \E \Big[ \big|E\big(H[S]\big) - E(B_S)\big| \Big] \\
    &\geq \E \Big[ \big|E\big(H[S]\big)\big| - |B_S| \Big] 
    = \frac{m}{2f} - \frac{|B|}{4f^2} \\
    &\geq \frac{m}{2f} - \frac{fm}{4f^2} = \frac{m}{4f} ~. && \text{(\Cref{lem:blocking-set})}
\end{align*}

On the other hand, by \Cref{obs:girth}, $\hat{H}_S$ has girth $\geq 2k+1$, and thus by the Moore bounds (\Cref{lem:Moore-bounds}), $\hat{m} = O(n^{1+1/k})$.
Hence,
\[
\frac{m}{4f} \leq \E[\hat{m}] = O(n^{1+1/k}) \quad \implies \quad m = O(fn^{1+1/k}).
\]
This concludes the proof of \Cref{thm:CFT-upper-bounds}(1).

\begin{remark}
    The proof works seamlessly for multi-graphs, and the same holds for all other proofs in this section.
\end{remark}

\subsection{The VCFT Setting}

We now analyze the size of the spanner produced by the greedy algorithm (\Cref{alg:greedy-spanner}) in the VCFT setting. As explained in \Cref{sect:technical_overview}, a na\"ive adaptation of the proof in the ECFT setting gives an $O(f^2n^{1+1/k})$ bound. 
To get our tight bound here, we go in a different route.

\subsubsection{Bipartite Moore Bounds}

Our analysis for the VCFT setting relies on Moore Bounds for bipartite graphs.
For our purposes, the strongest known bounds of Hoory~\cite{Hoory02} are not required, and the following weaker but simpler-to-prove bound is sufficient.

\begin{lemma}[Weak Bipartite Moore Bounds]\label{lem:bipartite-moore}
    Let $G$ be a bipartite graph with bipartition $V = L \cup R$, $|L|=x$, $|R| = y$.
    Suppose $G$ has girth $\geq 2k+1$. Then $|E(G)| = O(x^{1/k} y + x y^{1/k})$.
\end{lemma}
\begin{proof}
    Suppose $x \leq y$. We will prove that $|E(G)| = O(x^{1/k} y)$, which implies the lemma.
    
    Let $S \subseteq R$ be a uniformly random subset with $|S|=x$.
    The expected number of edges in the induced subgraph $G_S = G[L \cup S]$ is
    \[
    \sum_{e \in E(G)} \Pr[e \in L \times S] = |E(G)| \cdot \frac{x}{y} ~.
    \]
    Choose some $S$ exceeding the expectation, i.e., such that 
    \begin{equation}\label{eq:lower}
        |E(G_S)| \geq |E(G)| \cdot \frac{x}{y} ~.
    \end{equation}
    As $G_S$ also has girth $\geq 2k+1$, the Moore Bounds (\Cref{lem:Moore-bounds}) imply
    \begin{equation}\label{eq:upper}
        |E(G_S)| \leq O(|V(G_S)|^{1+1/k}) = O(x^{1+1/k}).
    \end{equation}
    Combining \Cref{eq:lower,eq:upper} yields $|E(G)| = O(x^{1/k} y)$, which concludes the proof.
\end{proof}

\subsubsection{Setting Up the Stage}
As before, let $H$ denote the output of the greedy spanner (\Cref{alg:greedy-spanner}), now in the VCFT setting.
Again, denote $m = |E(H)|$.
We now perform a ``clean-up'' phase for the sake of analysis, where we neglect some fraction of the edges of $H$, as follows.

First, we claim that $H$ contains at most $O(n^{1+1/k})$ \emph{monochromatic} edges.
Let $E_{mono} \coloneqq \{ \{ u,v\}\in E(H) \mid c(u)=c(v) \}$.
Note that $E_{mono}$ cannot contain cycles of $\leq 2k$ edges, as such cycles must contain some pair $(e=\{u,v\},c)$ of the blocking set $B$, but $|\{c,c(u),c(v)\}| \geq 2$. 
Hence, by the Moore Bounds (\Cref{lem:Moore-bounds}), $|E_{mono}| = O(n^{1+1/k})$.

From now on, we assume that $H$ has no monochromatic edges.
Namely, $H$ is henceforth $H-E_{mono}$.
This is justified as we can delete the $O(n^{1+1/k})$ edges of $E_{mono}$ (which is less than the upper bound we prove), while still maintaining a small blocking set (recall \Cref{remark:blocking-set-for-subgraphs}).

Next, for our analysis, we would like to ``make'' $H$ into a bipartite graph with the special property that all the vertices of any specific color class always lie on the same side.
To obtain this, we observe the following:

\begin{observation}\label{obs:bipartition}
    There is a partition $V(H) = L \cup R$, such that:
    \begin{itemize}
        \item[(i)] For every color $c$, all vertices with color $c$ are found only in one of $L,R$.
        \item[(ii)] The number of edges between $L$ and $R$ is at least $|E(H)| /2$.
    \end{itemize}
\end{observation}
\begin{proof}
    Independently for every color $c$, with probability $1/2$, all the $c$-colored vertices go to $L$, and otherwise they go to $R$.
    Each edge is split between $L,R$ with probability $1/2$ (as it is not monochromatic).
    Thus, the \emph{expected} number of $L$-to-$R$ edges is  $|E(H)|/2$, and some partition achieves $\geq |E(H)|/2$ many $L$-to-$R$ edges.
\end{proof}

Let $H'$ be the subgraph of $H$ obtained by deleting all edges that violate the bipartition.
Clearly, by (ii) of the last observation, it is enough to prove that $|E(H')| = O(f^{1-1/k} n^{1+1/k})$.
So from now on (justified again by \Cref{remark:blocking-set-for-subgraphs}), we assume that $H$ has only $L$-to-$R$ edges.

\subsubsection{Bounding the Size of $H$}

Again, denote $m = |E(H)|$.
The following is the key lemma of the analysis:
\begin{lemma}\label{lem:VCFT-upper-bound}
    Let $\hat{H}_S$ be a $p$-random blocked subgraph. Then $\E[\hat{m}] = O((pn)^{1+1/k})$.
\end{lemma}
\begin{proof}
    Let $L_S = \{v \in L \mid c(v) \in S\}$, i.e., the vertices from $L$ whose colors were sampled.
    Define $R_S$ similarly.
    The crux of the special color-bipartition is ensuring that $|L_S|$ and $|R_S|$ are \emph{independent} random variables.
    Recall that by \Cref{obs:girth}, $\hat{H}_S$ has girth $\geq 2k+1$.
    Using the bipartite Moore Bounds (\Cref{lem:bipartite-moore}), we now obtain
    \begin{align*}
        \E[\hat{m}]
        & \leq O\paren{ \E\sparen{ |L_S|^{1/k} \cdot |R_S| + |L_S| \cdot |R_S|^{1/k} } } \\
        &= O\paren{ \E\sparen{ |L_S|^{1/k}} \cdot \E\sparen{|R_S|} + \E\sparen{ |L_S| } \cdot \E\sparen{|R_S|^{1/k}} } && \text{(Independence)}\\
        &\leq O\paren{ \E[|L_S|]^{1/k} \cdot \E[|R_S|] + \E[|L_S|] \cdot \E[|R_S|]^{1/k} } && \text{(Jensen's inequality)} \\
        &= O\paren{ (p|L|)^{1/k} \cdot p|R| + p|L| \cdot (p|R|)^{1/k} }   \\
        &\leq O\paren{ (pn)^{1+1/k} } && 
    \end{align*}
    as required.
\end{proof}

We are now ready to complete the analysis.
Let $\hat{H}_S$ be a $p$-random blocked subgraph with $p = 1/(2f)$.
Then,
\[
\E[\hat{m}] \geq \E\sparen{ |E(H[S])|  - |B_S|}  = \frac{m}{4f^2} - \frac{|B|}{8f^3} \geq \frac{m}{8f^2} ~,
\]
where the last inequality is by \Cref{lem:blocking-set}.
Combining with \Cref{lem:VCFT-upper-bound}, we now obtain
\[
\frac{m}{8f^2} \leq \E[\hat{m}] = O\Big( \Big(\frac{n}{f} \Big)^{1+1/k} \Big) \quad \implies \quad m = O(f^{1-1/k} n^{1+1/k}).
\]
Thus, we have proved \Cref{thm:CFT-upper-bounds}(2).

\subsection{The MCFT Setting}

For the mixed setting, we suitably mix arguments from the ECFT and VCFT settings. We use the same paradigm of the VCFT proof: ``clean-up'', color-partition, and applying bipartite Moore Bounds. The only step of the paradigm that needs to be adapted is the ``clean-up'' phase of monochromatic edges, which we now prove:

\begin{lemma}
    $H$ has at most $O(fn^{1+1/k})$ edges $e = \{u,v\} \in E(H)$ with $c(u) = c(v)$.
\end{lemma}
\begin{proof}
    Let $E' = \{e = \{u,v\} \in E(H) \mid c(u) = c(v)\}$, and let $H' = (V(H), E')$.
    I.e., $H'$ is the subgraph with only the edges that have endpoints of the same color.
    Notice that each connected component of $H'$ contains vertices with only one color.
    
    Consider $B' = \{(e,c) \in B \mid e \in E'\}$.
    Each edge of $H'$ participates in $\leq f$ pairs from $B'$.
    We claim that $B'$ is a $2k$-blocking set for $H'$ in the \underline{\textbf{E}}CFT setting (i.e., when we ignore the existence of vertex colors).
    Condition (i) from \Cref{def:blocking-set} is clear.
    For Condition (ii), let $\Sigma$ be a cycle on $\leq 2k$ edges in $H'$.
    As $B$ is a $2k$-blocking set for $H$, there exists $(e,c) \in B$ such that the edge $e$ and the color $c$ (which does not damage $e$) both appear on $\Sigma$.
    All vertices on $\Sigma$, including the endpoints of $e$, have the same color.
    Thus, $c$ must appear on an \emph{edge} of $\Sigma$.
    Therefore, the pair $(e,c) \in B'$ certifies that Condition (ii) holds for $\Sigma$ in the \underline{\textbf{E}}CFT setting.

    Finally, as explained in \Cref{remark:blocking-set-proof}, our proof of \Cref{thm:CFT-upper-bounds}(1) showed that in the ECFT setting, an $n$-vertex graph that has a blocking set where each edge appears in at most $f$ pairs could only have $O(f n^{1+1/k})$ edges.
    Applying this to $H'$, we obtain that $|E'| = |E(H')| = O(f n^{1+1/k})$.
\end{proof}

By the above lemma (and \Cref{remark:blocking-set-for-subgraphs}), we can assume $H$ has no edges with endpoints of the same color.
For such $H$, we can apply the bipartitioning procedure as in the VCFT setting (ignoring the edge colors), and restrict the analysis to the resulting graph.
By applying \Cref{lem:VCFT-upper-bound} on the $p$-random blocked subgraph $\hat{H}_S$, we conclude that $\E[\hat{m}] = O\paren{ (pn)^{1+1/k} }$.

We set $p = 1/(2f)$.
Our final step, as before, is to give a lower bound on $\E[\hat{m}]$.
Let 
\begin{align*}
    E_2 &= \big\{e = \{u,v\} \in E(H) \ \mid \ |\{c(e),c(u),c(v)\}|= 2\big\} \\
    E_3 &= \big\{e = \{u,v\} \in E(H) \ \mid \ |\{c(e),c(u),c(v)\}|= 3\big\}
\end{align*}
As every edge has two different colors on its endpoints, $E(H) = E_2 \cup E_3$.
Next, let $B_2 = \{(e,x) \in B \mid e \in E_2\}$ and $B_3 = B-B_2$.
By \Cref{lem:blocking-set}, $|B_2| \leq f |E_2|$ and $|B_3| \leq f |E_3|$.
We now calculate:
\begin{align*}
    \E[\hat{m}] &\geq \E \sparen{ |E(H[S])| - |B_S|} \\
    &= \Big(\frac{|E_2|}{4f^2} - \frac{|B_2|}{8f^3} \Big) + \Big(\frac{|E_3|}{8f^3} - \frac{|B_3|}{16f^4} \Big) \\
    &\geq \frac{|E_2|}{8f^2} + \frac{|E_3|}{16f^3}
    \geq \frac{m}{16f^3}.
\end{align*}
Concluding, we get
\[
\frac{m}{16f^3} \leq \E[\hat{m}] = O\Big(\Big(\frac{n}{f} \Big)^{1+1/k}\Big)  \quad \implies \quad m = O(f^{2-1/k} n^{1+1/k}),
\]
which proves \Cref{thm:CFT-upper-bounds}(3).

\section{Lower Bounds}\label{sect:lower-bounds}

Our lower bounds are conditioned on the famous \emph{Erd\H{o}s Girth Conjecture} \cite{Erdos63}:
\begin{conjecture}[Erd\H{o}s Girth Conjecture]\label{conj:girth}
    For every integers $n,k \geq 1$, there exists an $n$-vertex graph with girth at least $2k+2$ and $\Omega (n^{1+1/k})$ edges.
\end{conjecture}
The Girth Conjecture is widely believed and commonly used as a basis for lower bounds on spanners, and particularly for the lower bounds of \cite{BodwinDPW18} on V/EFT-spanners. 

We start by proving a lower bound in the edge-colored case for \emph{simple} graphs (see \Cref{sect:multi-graphs} for a discussion on multi-graphs).

\begin{theorem}\label{thm:edge-color-spanner-lower-bound}
    Let $n, k, f \geq 1$ be integers.
    Assuming the Girth Conjecture, 
    there exists a \emph{simple} graph $G$ with $n$ vertices and colored edges,
    such that every $f$-ECFT $(2k-1)$-spanner of $G$ must have $\Omega(\min\{f n^{1+1/k}, n^2\})$ edges.
\end{theorem}

\begin{proof}
    Without loss of generality, we assume $f \leq \frac{1}{8} n^{1-1/k}$.
    If $f$ is larger, use the graph constructed for $\frac{1}{8} n^{1-1/k}$ color faults, which already gives the desired bound.

    Let $G_{ER}$ be the graph from the Girth Conjecture with $n$ vertices, i.e., it has girth $\geq 2k+2$ and $\alpha n^{1+1/k}$ edges, where $\alpha>0$ is a fixed constant.
    Without loss of generality, we may assume $\alpha \leq 1$ (otherwise, delete edges from $G_{ER}$).
    We prove that for every $1\leq i \leq f$, there exists a \emph{simple} $n$-vertex edge-colored graph $G_i$, with color palette $[i]$, such that the following conditions hold:
    \begin{itemize}
        \item[(i)] The number of edges in $G_i$ satisfies
        $
        i \cdot \frac{\alpha}{2} \cdot n^{1+1/k} \leq |E(G_i)| \leq i \cdot \alpha \cdot n^{1+1/k} .
        $ 
        \item[(ii)] For every color $c \in [i]$, there is no cycle of $\leq 2k+1$ $c$-colored edges in $G_i$.
    \end{itemize}
    The desired graph for the theorem is $G_f$, as it has $\Omega(fn^{1+1/k})$ edges by (i), and no proper subgraph which is an $f$ color fault-tolerant $(2k-1)$-spanner by (ii).
    
    The proof is by induction on $i$.
    The base case is just $G_1 = G_{ER}$ with all edges having the same color $1$.
    The induction step constructs $G_{i+1}$ from $G_i$, as follows.
    Consider a uniformly random bijection (or permutation)
    $\pi : V(G_{ER}) \to V(G_i)$.
    For every edge $\{u,v\} \in E(G_{ER})$,
    \[
    \Pr \big[ \{\pi(u),\pi(v)\} \in E(G_i) \big] = 
    \frac{|E(G_i)|}{\binom{n}{2}}
    \leq
    \frac{i \cdot \alpha n^{1+1/k}}{\frac{1}{4} n^2}
    \leq
    \frac{\frac{1}{8} n^{1-1/k} \cdot \alpha n^{1+1/k}}{\frac{1}{4} n^2}
    \leq
    \frac{1}{2}.
    \]
    Thus, in expectation, at least half of the edges in $G_{ER}$ are mapped by $\pi$ to \emph{non-edges} in $G_i$.
    Choose $\pi$ such that this last property holds. Define $G_{i+1}$ by adding these edges to $G_i$, and color them with the new color $i+1$. 
    The resulting graph $G_{i+1}$ is indeed simple, and
    properties (i) and (ii) follow immediately by the construction and the induction hypothesis.
\end{proof}

For the sake of completeness, we note that the lower bound of Bodwin et al. \cite{BodwinDPW18} for VFT spanners also applies to the more general VCFT setting.

\begin{theorem}[\cite{BodwinDPW18}]\label{thm:vertex-color-spanner-lower-bound}
    Let $n, k, f \geq 1$ be integers, $f \leq n$.
    Assuming the Girth Conjecture, there exists vertex-colored simple $n$-vertex graph $G$, such that every $f$-V(C)FT $(2k-1)$-spanner of $G$ must have $\Omega(f^{1-1/k} n^{1+1/k})$ edges.
\end{theorem}

Finally, we give the lower bound in the MCFT setting:

\begin{theorem}\label{thm:mixed-color-spanner-lower-bound}
    Let $n, k, f \geq 1$ be integers.
    Assuming the Girth Conjecture, 
    there exists a \emph{simple} graph $G$ with $n$ vertices, where both edges and vertices are colored, such that every $f$-MCFT $(2k-1)$-spanner of $G$ must have $\Omega(\min\{f^{2-1/k} n^{1+1/k}, n^2\})$ edges.
\end{theorem}

\begin{proof}
    Let $G_{EC}$ be the edge-colored $n$-vertex graph of \Cref{thm:edge-color-spanner-lower-bound}, 
    with $\Omega(\min\{fn^{1+1/k},n^2\})$ edges and no proper $f$-ECFT $(2k-1)$-spanner.
    Denote its color palette by $\C_E$.
    We may also assume that $G_{EC}$ is bipartite, as every graph can be made bipartite by deleting at most half of its edges.%
    \footnote{
        This can be seen by considering a random bipartition.
    }
    Let $V(G_{EC}) = L \cup R$ be the partition to left and right sides.

    We construct a graph $G$ with both edge and vertex colors, as follows.
    The vertex set is
    \[
    V(G) = (L \times [f]) \cup (R \times [f]), 
    \]
    i.e., there are $f$ copies of each vertex from $V(G) = L \cup R$.
    There are two sets of new colors for the vertices (disjoint from each other and from $\C_E$):
    \[
    \C_L = \{\ell_1, \dots, \ell_f\} \quad \text{and} \quad \C_R = \{r_1, \dots, r_f\} .
    \]
    A vertex $(u,i) \in L \times [f]$ gets the color $\ell_i$.
    Similarly, a vertex $(v,j) \in R \times [f]$ gets the color $r_j$.
    We then place a copy of $G_{EC}$ on $(L \times \{i\}) \cup (R \times \{j\})$ for every $i,j \in [f]$. Formally, the
     edge set is defined by
    \[
    E(G) = \big\{ \{(u,i), (v,j)\} \mid i,j \in [f], ~ u \in L,~ v\in R,~ \{u,v\} \in E(G_{EC}) \big\}.
    \]
    Every edge $\{(u,i), (v,j)\} \in E(G)$ is assigned the color that is given in $G_{EC}$ to the corresponding edge $\{u,v\}$ (from the palette $\C_E$).
    Note that the copies are edge-disjoint, so the resulting $G$ is simple.

    We now show that $G$ has no $3f$-MCFT $(2k-1)$-spanner but itself.
    Indeed, suppose $H$ is such a spanner.
    Choose any $i,j \in [f]$. Let $F_L^i = \C_L - \{\ell_i\}$ and $F_R^j = \C_R - \{r_j\}$.
    Consider the graph $H - (F_L^i \cup F_R^j)$.
    As $|F_L^i \cup F_R^j| \leq 2f$, this graph must be an $f$-MCFT (and, as $(\C_L \cup \C_R) \cap \C_E = \emptyset$, also $f$-ECFT) $(2k-1)$-spanner of $G-(F_L^i \cup F_R^j)$.
    However, $G-(F_L^i \cup F_R^j)$ is just the copy of $G_{EC}$ on $(L \times \{i\}) \cup (R \times \{j\})$, which has no such spanner but itself.
    Therefore, each of the copies of $G_{EC}$ must be contained in $H$, hence $H = G$ as claimed.

    Finally, we have that $|V(G)| = f |V(G_{EC})| = fn$ and
    \[
    |E(G)| = f^2 |E(G_{EC})| = f^2 \cdot \Omega(\min\{fn^{1+1/k}, n^2\}) = \Omega(\min\{f^{2-1/k} (fn)^{1+1/k}, (fn)^2)\}).
    \]
    The result follows by rescaling the proof: $n \gets fn$ and $f \gets 3f$.
\end{proof}

\section{Efficient Constructions}\label{sect:efficient-constructions}

We now show an efficient (polynomial time) construction yielding E/V/MCFT spanners whose sizes match the bounds of \Cref{sect:upper-bounds} up to  a $\poly(k)$ factor:

\begin{theorem}\label{thm:modified-greedy}
    There is a polynomial time algorithm for constructing $f$-CFT $(2k-1)$-spanners, yielding spanner with size $s$ such that:
    \begin{enumerate}
        \item $s = O(k f n^{1+1/k})$ in the ECFT setting,
        \item $s = O(k f^{1-1/k} n^{1+1/k})$ in the VCFT setting, and
        \item $s = O(k^2 f^{2-1/k} n^{1+1/k})$ in the MCFT setting.
    \end{enumerate}
    The running time is $O(mf \cdot s)$.
\end{theorem}

This is by modifying the ``Modified FT Greedy Spanner'' algorithm of Dinitz and Robelle \cite{DinitzR20} to the colored settings.
As in \cite{DinitzR20}, we start by considering the case where $G$ is \emph{unweighted}.
Assume this from now on; we discuss how weights are handled at the end of this section.

\paragraph{The High-Level Idea.}
Consider the beginning of some iteration of the FT greedy algorithm (\Cref{alg:greedy-spanner}), which examines $e = \{u,v\} \in E(G)$.
The intractable part of the algorithm is trying to find a fault-set $F_{e}$ of size $\leq f$ such that $\dist_{H -F_{e}} (u,v) > 2k-1$.
If no such set is found, $e$ can clearly be discarded.
In the complementary case, $e$ is added to $H$, and the set $F_e$ serves to account for $\leq f$ pairs in the $2k$-blocking set constructed by \Cref{lem:blocking-set}, which ultimately yields the size bound on the final $H$.

Dinitz and Robelle's idea is replacing this by tractable \emph{approximate} subroutine, that settles for finding a set $\tilde{F}_e$ with size $O(kf)$, and ensures that $e$ can be discarded if such a set is not found.
This subroutine is based on \emph{another greedy approach}, of computing a maximal set of `disjoint' (in some sense) $u$-$v$ paths of length $\leq 2k-1$ in $H$. The ``blame set'' $\tilde{F}_e$ is found if the maximal set of paths
contains at most $f$ paths, in which case $e$ is inserted to the spanner.
If there are more than $f$ `disjoint' paths, $e$ can be safely discarded.
Dinitz and Robelle~\cite{DinitzR20} consider the VFT setting, where `disjoint' means internally vertex-disjoint, and the set $\tilde{F}_e$ contains the internal vertices of those paths.
This approach can easily be adapted to the colored setting, as we now explain.

\paragraph{Subroutine for Deciding on an Edge.}
The subroutine 
works in $f+1$ iterations.
We preserve the invariant that if we completed iteration $i$ without halting, then we have $i$ $u$-$v$ paths $\{P_j\}_{j=1}^{i}$ in $H$, each of length $\leq 2k-1$.
Furthermore, let $\C_j$ be the set of colors appearing on $P_j$ that do not damage $e$.
Then the color sets $\{\C_j\}_{j=1}^{i}$ are mutually disjoint.
To execute iteration $i$, we compute the shortest $u$-$v$ path $P_i$ in $H - \bigcup_{j=1}^{i-1} \C_j$.
If $P_i$ is of length $\geq 2k$, we halt by deciding to include $e$, and `blame' the set $\tilde{F}_{e} = \bigcup_{j=1}^{i-1} \C_j$ for this decision.
Crucially, we can bound $|\tilde{F}_{e}| \leq 4kf$ (as a path of length $\leq 2k-1$ can have at most $(2k-1) + 2k$ vertices and edges, and thus at most $4k$ different colors).
Otherwise, if $P_i$ has length $\leq 2k-1$, we add $P_i$ and continue to the next iteration.
If the last iteration $f+1$ was completed, 
we have found $f+1$ $u$-$v$ paths of length $\leq 2k-1$, of which every pair may share only the colors that damage $e$.
This means that no set of at most $\leq f$ color faults (that does not damage $e$ itself) can cause the distance between $u,v$ in $H$ to be $> 2k-1$.
In this case, we call $e$ \emph{replaceable}, and we safely discard it while ensuring the correctness of the final output spanner.

\medskip
The pseudo-code for the modified FT greedy algorithm using this subroutine appears as \Cref{alg:modified-greedy-spanner-unweighted}.
The running time analysis is identical to \cite{DinitzR20}.
\begin{algorithm}[H]
\caption{$\mathsf{ModifiedCFTGreedySpanner}(G,k,f)$}\label{alg:modified-greedy-spanner-unweighted}
(Assuming $G$ is unweighted)
\begin{algorithmic}[1]
    \State $H \gets (V(G), \emptyset)$
    \For{each $e = \{u,v\} \in E(G)$ in arbitrary order}
        \If{$e$ is not replaceable because of color set $\tilde{F}_e$}
            \State $H \gets H \cup \{e\}$
        \EndIf
    \EndFor
    \State \Return $H$
\end{algorithmic}
\end{algorithm}

\paragraph{Size Analysis.}
The size analysis requires only minor modifications compared to the one for the (exponential-time) FT greedy algorithm, presented in \Cref{sect:upper-bounds}.
We define
\[
\tilde{B} = \{(e,c) \mid e\in E(H), c \in \tilde{F}_e\}.
\]
Since right before an edge $e = \{u,v\}$ is inserted into $H$ we have that $\dist_{H - \tilde{F}_e} (u,v) > 2k-1$, then by the same proof as in \Cref{lem:blocking-set},  $\tilde{B}$ is a $2k$-blocking set.
We can now repeat the analysis via random blocked subgraphs, with the following modifications:
\begin{itemize}
    \item[(1)] The definition of a $p$-random blocked subgraph (\Cref{def:random-blocked-subgraph}) uses $\tilde{B}$ instead of $B$.
    \item[(2)] We set $p = \frac{1}{8kf}$ instead of $\frac{1}{2f}$, to account for the fact that $|\tilde{B}| \leq 4kf m$, instead of $|B| \leq fm$.
\end{itemize}

The resulting bounds on the size of the spanner are then bigger than the ones in \Cref{thm:CFT-upper-bounds} by a factor of $O(k)$ in the ECFT and the VCFT settings, and by a factor of $O(k^2)$ in the MCFT setting, proving \Cref{thm:modified-greedy} for the unweighted case.

\paragraph{Handling Weights.}
The way \cite{{DinitzR20}} handle a weighted graph $G$ works `as is' also in our case.
It is surprisingly simple: 
Run the \emph{unweighted} \Cref{alg:modified-greedy-spanner-unweighted} on $G$, (i.e., ignoring the weights), but with the ``correct'' increasing order of weight.
The correctness is by the following observation: 
If $e = \{u,v\} \in G$ was found replaceable, then at that time, for every set $F$ of at most $\leq f$ faults (not damaging $e$), $H-F$ contains a $u$-$v$ path of at most $2k-1$ edges.
All these edges precede $e$ in the ordering, so the total weight of the path is $\leq (2k-1)w(e)$.

\section{Conclusion}

In this work, we introduced the concept of color fault-tolerant (CFT) spanners, which generalize vertex/edge fault-tolerant spanners.
We give tight bounds on the size required by CFT spanners, and polynomial time constructions of such spanners with nearly optimal size guarantees.
We suggest a few possible directions for future research:

\begin{itemize}
    \item Can we provide faster algorithms for CFT spanners with size guarantees close to those of the FT greedy algorithm?
    See e.g.\ the construction by Parter \cite{Parter22} in the VFT setting, based on the Baswana-Sen algorithm~\cite{BaswanaS07}.

    \item As noted, the min-cut max-flow theorem does not hold in the colored setting. Hence, one may consider investigating a different notion from the ``color cuts'' of CFT-connectivity certificates, which pertains to ``colorful flow''. That is, a subgraph $H$ of $G$ such that every two vertices $u,v$ have $\lambda$ color-disjoint paths in $G$ iff they have $\lambda$ such paths in $H$.

\end{itemize}

\ifDRAFT
\else
\paragraph{Acknowledgments.}
We are grateful to Merav Parter for encouraging this collaboration, and for helpful guidance and discussions.
We thank Nathan Wallheimer for useful discussions.
\fi

\bibliographystyle{alphaurl}
\bibliography{references.bib}

\newcommand{\etalchar}[1]{$^{#1}$}
\begin{thebibliography}{BDPW18}

\bibitem[ADD{\etalchar{+}}93]{AlthoferDDJS93}
Ingo Alth{\"{o}}fer, Gautam Das, David~P. Dobkin, Deborah Joseph, and Jos{\'{e}} Soares.
\newblock On sparse spanners of weighted graphs.
\newblock {\em Discret. Comput. Geom.}, 9:81--100, 1993.
\newblock \href {https://doi.org/10.1007/BF02189308} {\path{doi:10.1007/BF02189308}}.

\bibitem[BDPW18]{BodwinDPW18}
Greg Bodwin, Michael Dinitz, Merav Parter, and Virginia~Vassilevska Williams.
\newblock Optimal vertex fault tolerant spanners (for fixed stretch).
\newblock In {\em Proceedings of the Twenty-Ninth Annual {ACM-SIAM} Symposium on Discrete Algorithms, {SODA}}, pages 1884--1900, 2018.
\newblock \href {https://doi.org/10.1137/1.9781611975031.123} {\path{doi:10.1137/1.9781611975031.123}}.

\bibitem[BDR21]{BodwinDR21}
Greg Bodwin, Michael Dinitz, and Caleb Robelle.
\newblock Optimal vertex fault-tolerant spanners in polynomial time.
\newblock In {\em Proceedings of the 2021 {ACM-SIAM} Symposium on Discrete Algorithms, {SODA}}, pages 2924--2938, 2021.
\newblock \href {https://doi.org/10.1137/1.9781611976465.174} {\path{doi:10.1137/1.9781611976465.174}}.

\bibitem[BDR22]{BodwinDR22}
Greg Bodwin, Michael Dinitz, and Caleb Robelle.
\newblock Partially optimal edge fault-tolerant spanners.
\newblock In {\em Proceedings of the 2022 {ACM-SIAM} Symposium on Discrete Algorithms, {SODA}}, pages 3272--3286, 2022.
\newblock \href {https://doi.org/10.1137/1.9781611977073.129} {\path{doi:10.1137/1.9781611977073.129}}.

\bibitem[BP19]{BodwinP19}
Greg Bodwin and Shyamal Patel.
\newblock A trivial yet optimal solution to vertex fault tolerant spanners.
\newblock In {\em Proceedings of the 2019 {ACM} Symposium on Principles of Distributed Computing, {PODC}}, pages 541--543, 2019.
\newblock \href {https://doi.org/10.1145/3293611.3331588} {\path{doi:10.1145/3293611.3331588}}.

\bibitem[BS04]{BaswanaS04}
Surender Baswana and Sandeep Sen.
\newblock Approximate distance oracles for unweighted graphs in $\tilde{O}(n^2)$ time.
\newblock In {\em Proceedings of the Fifteenth Annual {ACM-SIAM} Symposium on Discrete Algorithms, {SODA}}, pages 271--280, 2004.
\newblock URL: \url{http://dl.acm.org/citation.cfm?id=982792.982830}.

\bibitem[BS07]{BaswanaS07}
Surender Baswana and Sandeep Sen.
\newblock A simple and linear time randomized algorithm for computing sparse spanners in weighted graphs.
\newblock {\em Random Struct. Algorithms}, 30(4):532--563, 2007.
\newblock \href {https://doi.org/10.1002/rsa.20130} {\path{doi:10.1002/rsa.20130}}.

\bibitem[CDP{\etalchar{+}}07]{CoudertDPRV07}
David Coudert, Pallab Datta, Stephane Perennes, Herv{\'{e}} Rivano, and Marie{-}Emilie Voge.
\newblock Shared risk resource group complexity and approximability issues.
\newblock {\em Parallel Process. Lett.}, 17(2):169--184, 2007.
\newblock \href {https://doi.org/10.1142/S0129626407002958} {\path{doi:10.1142/S0129626407002958}}.

\bibitem[CKT93]{CheriyanKT93}
Joseph Cheriyan, Ming{-}Yang Kao, and Ramakrishna Thurimella.
\newblock Scan-first search and sparse certificates: An improved parallel algorithms for k-vertex connectivity.
\newblock {\em {SIAM} J. Comput.}, 22(1):157--174, 1993.

\bibitem[CLPR10]{ChechikLPR10}
Shiri Chechik, Michael Langberg, David Peleg, and Liam Roditty.
\newblock Fault tolerant spanners for general graphs.
\newblock {\em {SIAM} J. Comput.}, 39(7):3403--3423, 2010.
\newblock \href {https://doi.org/10.1137/090758039} {\path{doi:10.1137/090758039}}.

\bibitem[Cow01]{Cowen01}
Lenore Cowen.
\newblock Compact routing with minimum stretch.
\newblock {\em J. Algorithms}, 38(1):170--183, 2001.
\newblock \href {https://doi.org/10.1006/jagm.2000.1134} {\path{doi:10.1006/jagm.2000.1134}}.

\bibitem[CW04]{CowenW04}
Lenore Cowen and Christopher~G. Wagner.
\newblock Compact roundtrip routing in directed networks.
\newblock {\em J. Algorithms}, 50(1):79--95, 2004.
\newblock \href {https://doi.org/10.1016/j.jalgor.2003.08.001} {\path{doi:10.1016/j.jalgor.2003.08.001}}.

\bibitem[DK11]{DinitzK11}
Michael Dinitz and Robert Krauthgamer.
\newblock Fault-tolerant spanners: better and simpler.
\newblock In {\em Proceedings of the 30th Annual {ACM} Symposium on Principles of Distributed Computing, {PODC}}, pages 169--178, 2011.
\newblock \href {https://doi.org/10.1145/1993806.1993830} {\path{doi:10.1145/1993806.1993830}}.

\bibitem[DR20]{DinitzR20}
Michael Dinitz and Caleb Robelle.
\newblock Efficient and simple algorithms for fault-tolerant spanners.
\newblock In {\em {PODC} '20: {ACM} Symposium on Principles of Distributed Computing}, pages 493--500, 2020.
\newblock \href {https://doi.org/10.1145/3382734.3405735} {\path{doi:10.1145/3382734.3405735}}.

\bibitem[EBR{\etalchar{+}}03]{EGBERRL03}
Georgios Ellinas, Eric Bouillet, Ramu Ramamurthy, J.-F Labourdette, Sid Chaudhuri, and Krishna Bala.
\newblock Routing and restoration architectures in mesh optical networks.
\newblock {\em Opt Networks Mag}, 4, 01 2003.

\bibitem[EEST08]{ElkinEST08}
Michael Elkin, Yuval Emek, Daniel~A. Spielman, and Shang{-}Hua Teng.
\newblock Lower-stretch spanning trees.
\newblock {\em {SIAM} J. Comput.}, 38(2):608--628, 2008.
\newblock \href {https://doi.org/10.1137/050641661} {\path{doi:10.1137/050641661}}.

\bibitem[Erd63]{Erdos63}
Paul Erd\H{o}s.
\newblock Extremal problems in graph theory.
\newblock {\em Theory of Graphs and its Applications}, 1963.

\bibitem[FNY{\etalchar{+}}20]{ForsterNYSY20}
Sebastian Forster, Danupon Nanongkai, Liu Yang, Thatchaphol Saranurak, and Sorrachai Yingchareonthawornchai.
\newblock Computing and testing small connectivity in near-linear time and queries via fast local cut algorithms.
\newblock In {\em Proceedings of the 2020 {ACM-SIAM} Symposium on Discrete Algorithms, {SODA}}, pages 2046--2065, 2020.

\bibitem[GK13]{GhaffariK13}
Mohsen Ghaffari and Fabian Kuhn.
\newblock Distributed minimum cut approximation.
\newblock In {\em Distributed Computing - 27th International Symposium, {DISC}}, volume 8205 of {\em Lecture Notes in Computer Science}, pages 1--15. Springer, 2013.

\bibitem[Hoo02]{Hoory02}
Shlomo Hoory.
\newblock The size of bipartite graphs with a given girth.
\newblock {\em J. Comb. Theory, Ser. {B}}, 86(2):215--220, 2002.
\newblock \href {https://doi.org/10.1006/JCTB.2002.2123} {\path{doi:10.1006/JCTB.2002.2123}}.

\bibitem[Hu03]{Hu03}
Jian~Qiang Hu.
\newblock Diverse routing in optical mesh networks.
\newblock {\em IEEE Transactions on Communications}, 51(3):489--494, 2003.
\newblock \href {https://doi.org/10.1109/TCOMM.2003.809779} {\path{doi:10.1109/TCOMM.2003.809779}}.

\bibitem[KM97]{KargerM97}
David~R. Karger and Rajeev Motwani.
\newblock An {NC} algorithm for minimum cuts.
\newblock {\em {SIAM} J. Comput.}, 26(1):255--272, 1997.

\bibitem[Kui12]{Kuipers12}
F.~Kuipers.
\newblock An overview of algorithms for network survivability.
\newblock {\em ISRN Communications and Networking}, 2012, 12 2012.
\newblock \href {https://doi.org/10.5402/2012/932456} {\path{doi:10.5402/2012/932456}}.

\bibitem[LNP{\etalchar{+}}21]{LiNPSY21}
Jason Li, Danupon Nanongkai, Debmalya Panigrahi, Thatchaphol Saranurak, and Sorrachai Yingchareonthawornchai.
\newblock Vertex connectivity in poly-logarithmic max-flows.
\newblock In {\em {STOC} '21: 53rd Annual {ACM} {SIGACT} Symposium on Theory of Computing}, pages 317--329, 2021.

\bibitem[Mat93]{Matula93}
David~W. Matula.
\newblock A linear time $2+\epsilon$ approximation algorithm for edge connectivity.
\newblock In {\em Proceedings of the Fourth Annual {ACM/SIGACT-SIAM} Symposium on Discrete Algorithms, {SODA}}, pages 500--504, 1993.

\bibitem[MN02]{MENTA02}
Eytan~H. Modiano and Aradhana Narula{-}Tam.
\newblock Survivable lightpath routing: a new approach to the design of wdm-based networks.
\newblock {\em {IEEE} J. Sel. Areas Commun.}, 20(4):800--809, 2002.
\newblock \href {https://doi.org/10.1109/JSAC.2002.1003045} {\path{doi:10.1109/JSAC.2002.1003045}}.

\bibitem[NI92]{NagamochiI92}
Hiroshi Nagamochi and Toshihide Ibaraki.
\newblock A linear-time algorithm for finding a sparse k-connected spanning subgraph of a k-connected graph.
\newblock {\em Algorithmica}, 7(5{\&}6):583--596, 1992.
\newblock \href {https://doi.org/10.1007/BF01758778} {\path{doi:10.1007/BF01758778}}.

\bibitem[Par19]{Par19}
Merav Parter.
\newblock Small cuts and connectivity certificates: A fault tolerant approach.
\newblock In {\em 33rd International Symposium on Distributed Computing, {DISC}}, volume 146 of {\em Leibniz International Proceedings in Informatics (LIPIcs)}, pages 30:1--30:16. Schloss Dagstuhl--Leibniz-Zentrum fuer Informatik, 2019.
\newblock \href {https://doi.org/10.4230/LIPIcs.DISC.2019.30} {\path{doi:10.4230/LIPIcs.DISC.2019.30}}.

\bibitem[Par22]{Parter22}
Merav Parter.
\newblock Nearly optimal vertex fault-tolerant spanners in optimal time: sequential, distributed, and parallel.
\newblock In {\em {STOC} '22: 54th Annual {ACM} {SIGACT} Symposium on Theory of Computing}, pages 1080--1092, 2022.
\newblock \href {https://doi.org/10.1145/3519935.3520047} {\path{doi:10.1145/3519935.3520047}}.

\bibitem[PS89]{PelegS89}
David Peleg and Alejandro~A. Sch{\"{a}}ffer.
\newblock Graph spanners.
\newblock {\em J. Graph Theory}, 13(1):99--116, 1989.
\newblock \href {https://doi.org/10.1002/jgt.3190130114} {\path{doi:10.1002/jgt.3190130114}}.

\bibitem[PU88]{PelegU88}
David Peleg and Eli Upfal.
\newblock A tradeoff between space and efficiency for routing tables.
\newblock In {\em Proceedings of the 20th Annual {ACM} Symposium on Theory of Computing, {STOC}}, pages 43--52, 1988.
\newblock \href {https://doi.org/10.1145/62212.62217} {\path{doi:10.1145/62212.62217}}.

\bibitem[PU89]{PelegU89a}
David Peleg and Jeffrey~D. Ullman.
\newblock An optimal synchronizer for the hypercube.
\newblock {\em {SIAM} J. Comput.}, 18(4):740--747, 1989.
\newblock \href {https://doi.org/10.1137/0218050} {\path{doi:10.1137/0218050}}.

\bibitem[TZ01]{ThorupZ01}
Mikkel Thorup and Uri Zwick.
\newblock Compact routing schemes.
\newblock In {\em Proceedings of the Thirteenth Annual {ACM} Symposium on Parallel Algorithms and Architectures, {SPAA}}, pages 1--10, 2001.
\newblock \href {https://doi.org/10.1145/378580.378581} {\path{doi:10.1145/378580.378581}}.

\bibitem[TZ05]{ThorupZ05}
Mikkel Thorup and Uri Zwick.
\newblock Approximate distance oracles.
\newblock {\em J. {ACM}}, 52(1):1--24, 2005.
\newblock \href {https://doi.org/10.1145/1044731.1044732} {\path{doi:10.1145/1044731.1044732}}.

\bibitem[ZCTZ11]{ZPT11}
Peng Zhang, Jin-Yi Cai, Lin-Qing Tang, and Wen-Bo Zhao.
\newblock Approximation and hardness results for label cut and related problems.
\newblock {\em Journal of Combinatorial Optimization}, 21:192--208, 02 2011.
\newblock \href {https://doi.org/10.1007/s10878-009-9222-0} {\path{doi:10.1007/s10878-009-9222-0}}.

\end{thebibliography}

\appendix
\section{FT Spanners for Multi-Graphs}\label{sect:multi-graphs}

All of the results in this paper apply also when $G$ is a colored \emph{multi}-graph, i.e., when many parallel edges between two vertices are allowed.
Especially in edge-colored settings, this seems like a rather natural assumption in light of the real-world motivation of communication networks.
However, it is presumably less natural in non-colored settings, so explicit discussions on multi-graphs have mostly been neglected in previous work on FT spanners.
We next provide such a short discussion, which shows an interesting phenomenon: the bounds for uncolored multi-graphs are the same as for colored simple graphs.

\paragraph{EFT: $\Theta(fn^{1+1/k})$ size bounds.}
Our upper bound for ECFT-spanners in \Cref{thm:CFT-upper-bounds}(1) applies for multi-graphs, and EFT is just a special case.
Additionally, we note that the EFT-spanner algorithm  of Chechik et al.\ \cite{ChechikLPR10} also works on multi-graphs, and achieves this size bound.
For the lower bound, one may take the Girth Conjecture graph with $\Omega(n^{1+1/k})$ edges, and give each of its edges multiplicity $f$.
Clearly, such a graph cannot have a proper $f$-EFT $(2k-1)$-spanner.

\paragraph{VFT: $\Theta(f^{1-1/k}n^{1+1/k})$ size bounds.}
In the VFT setting, parallel edges can be discarded by keeping the edge of lowest weight among them. This does not change distances (regardless of which vertex faults occur).
Hence, the bounds are the same as for simple graphs.

\paragraph{MFT: $\Theta(f^{2-1/k}n^{1+1/k})$ size bounds.}
Constructing a lower bound instance is easy: Take the VFT lower bound instance which is a simple graph with $\Omega(f^{1-1/k}n^{1+1/k})$ edges and replace each edge by $f$ parallel edges.
Our upper bound for MCFT spanners applies also here.

\section{Color Lists}\label{sect:color-lists}
This section is devoted to proving \Cref{thm:color-lists}.
We start with a formal description of the color lists setting.
Here, each edge/vertex $x \in E(G) \cup V(G)$ is associated with a list $\ell(x) \subseteq \C$, where $\C$ is a finite color palette.
The failure of any color $c \in \ell(x)$ causes $x$ to crash.
Formally, for $c \in \C$, we now denote $E_c(G) = \{e \in E(G) \mid c \in \ell(e)\}$ and $V_c (G) = \{v \in V(G) \mid c \in \ell(v)\}$.
The fault universe is $\F = \{E_c (G) \cup V_c (G)\}_{c \in \C}$.
Again, the elements in $\F$ are in 1-to-1 correspondence to the colors in $\C$ (that is, we have faults of colors).
We say $s \in \F$ damages $e = \{u,v\} \in E(G)$ if $s \cap (\ell(e) \cup \ell(u) \cup \ell(v)) \neq \emptyset$.
As $s$ is identified with some color $c \in \C$, this means that the color $c$ damages $e$ if it appears in $\ell(e) \cup \ell(u) \cup \ell(v)$.
For $S\subseteq \F$ (or $S \subseteq \C$), the notations $G-S$ and $G[S]$ are defined as in \Cref{sect:prelim}.
In particular, note that a vertex $v \in V(G)$ such that $\ell(v) \not \subseteq S$ becomes isolated in $G[S]$.
We say that $G$ is \emph{$(\mu,\nu)$-list-colored} if $|\ell(e)| \leq \mu$ for all $e \in E(G)$, and $|\ell(v)| \leq \nu$ for all $v \in V(G)$.
Namely, edge (resp., vertex) lists are of size at most $\mu$ (resp., $\nu$).
Letting $(\mu,\nu) = (1,0)$, $(0,1)$ or $(1,1)$ gives the ECFT, VCFT or MCFT settings, respectively.
Throughout, we will consider a $(\mu,\nu)$-list-colored graph $G$, where $\mu,\nu$ are arbitrary fixed integer constants.

\subsection{Upper Bound}

As in \Cref{sect:upper-bounds}, the upper bound on the size of an $f$-CFT $(2k-1)$-spanner is obtained by analyzing the output spanner $H$ of the FT-greedy algorithm (\Cref{alg:greedy-spanner}).
Observe that \Cref{lem:blocking-set} still holds: $H$ has a blocking set where each edge $e \in E(H)$ appears in at most $f$ pairs $(e,x)\in B$.
Again, this property alone suffices for proving the size upper bound.
Namely, we have:

\begin{theorem}\label{thm:color-lists-upper-bound}
    Let $H$ be a $(\mu,\nu)$-list-colored graph with $n$ vertices.
    Suppose there exists a $2k$-blocking set $B$ for $H$, such that for every $e \in E(H)$, there are at most $f$ pairs of the form $(e,c) \in B$.
    Then $|E(H)| = O(f^{\mu + \nu(1-1/k)} n^{1+1/k})$.
\end{theorem}

Without loss of generality, we assume that the edge (vertex) color-lists in $H$ are all of size \emph{exactly $\mu$} (\emph{exactly $\nu$}).
Otherwise, ``pad" the lists with new colors, which does not hurt the blocking set condition.
Throughout, we denote $m = |E(H)|$.

\medskip
The proof of \Cref{thm:color-lists} is by induction on $\nu$, the size of the vertex lists.
As in \Cref{sect:upper-bounds}, we analyze the expection of $\hat{m}$ number of edges in a $p$-random blocked subgraph $\hat{H}_S$ (as defined in \Cref{def:random-blocked-subgraph}), with $p = 1/(2f)$.
Note that \Cref{obs:girth} stays true, i.e., a $p$-random blocked subgraph always has girth $> 2k$.

\paragraph{Base Case: $\nu = 0$.}
By the Moore Bounds (\Cref{lem:Moore-bounds}), $\E [\hat{m}] = O(n^{1+1/k})$.
On the other hand, note that an edge $e \in E(H)$ has probability $(2f)^{-\mu}$ to survive in $H[S]$. Thus,
\begin{align*}
    \E[\hat{m}] 
    \geq \E \Big[ \big|E\big(H[S]\big)\big| - |B_S| \Big] 
    = \frac{m}{(2f)^\mu} - \frac{|B|}{(2f)^{\mu + 1}}
    \geq \frac{m}{2^{\mu + 1} f^\mu}.
\end{align*}
Combining, we get $m \leq 2^{\mu + 1} f^\mu \cdot O(n^{1+1/k}) = O(f^\mu n^{1+1/k})$, as needed.\footnote{Recall that $\mu,\nu$ are constants.}

\paragraph{Induction Step.}
For $\nu > 0$, we follow the strategy of the MCFT setting from \Cref{sect:upper-bounds}.
We start with the ``clean-up" phase.

\begin{lemma}
    Let $E_{\conf} = \{ e = \{u,v\} \in E(H) \mid \ell(u) \cap \ell(v) \neq \emptyset\}$ be the set of \emph{conflicting} edges, i.e., edges whose endpoints share a color in their lists.
    Then
    $|E_{\conf}| = O(f^{\mu+(\nu-1)(1-1/k)} \cdot n^{1+1/k})$.
\end{lemma}
\begin{proof}
    For $c \in \C$, let $H_c$ be the subgraph of $H$ induced on $V_c (H)$, i.e., on those vertices with color $c$ in their lists.
    Let $n_c = |V_c (H)|$.
    Let $B_c$ the $2k$-blocking set of $H_c$ induced from $B$, i.e., $B_c = \{(e,x) \in B \mid e \in E(H_c)\}$.
    As $c$ damages every edge in $E(H_c)$, it cannot appear in $B_c$.
    Thus, we may omit $c$ from the vertex lists of $H_c$, and $B_c$ remains a $2k$-blocking set for this $(\mu, \nu-1)$-list-colored graph.
    So, by induction, $|E(H_c)| = O(f^{\mu+(\nu-1)(1-1/k)} n_c^{1+1/k})$.
    As $E_{\conf} \subseteq \bigcup_{c} E(H_c)$, we get
    $
    |E_{\conf}| \leq \sum_{c} |E(H_c)| = O\paren{  f^{\mu+(\nu-1)(1-1/k)} \cdot \sum_{c} n_c^{1+1/k}}.
    $
    Note that
    $
    \sum_{c} n_c^{1+1/k} \leq n^{1/k} \cdot \sum_{c} n_c = n^{1/k} \cdot \nu n = O(n^{1+1/k}),
    $
    so the lemma follows.
\end{proof}
By the above lemma, we may assume, without loss of generality, that $H$ has no conflicting edges (i.e., $H$ is henceforth replaced with $H-E_{\conf}$).
The next step is the color-bipartition.

\begin{observation}    
    There is a subgraph $H'$ of $H$ with the following properties:
    \begin{itemize}
        \item[(i)] $H'$ is bipartite with partition $V(H') = L \cup R$.
        \item[(ii)] For each color $c$, all vertices of $H'$ that have $c$ in their lists appear on the same side of $H'$, namely, either in $L$ or in $R$.
        \item[(iii)] $|E(H')| \geq 2^{-2\nu + 1} \cdot m$.
    \end{itemize}
\end{observation}
\begin{proof}
    Independently for each $c \in C$, with probability $1/2$, choose the side $L$ for $c$, otherwise choose $R$.
    A vertex $v \in V(H)$ is deleted if the colors in $\ell(v)$ did not all agree on the same side.
    Otherwise, it goes to the agreed side in $H'$.
    The probability of each edge $e \in E(H)$ to have one endpoint in $L$ and the other in $R$, and thus survive in $H'$, is $2^{-2\nu+1}$ (as it is not conflicting).
    Thus, $\E[ |E(H')| ] = 2^{-2\nu+1} \cdot m$, which implies the claim.
\end{proof}

As $2^{-2\nu+1} = \Theta(1)$, we may replace $H$ with $H'$ from now on.
I.e., henceforth we assume, without loss of generality, that $H$ has a bipartition $V(H) = L \cup R$ where for each color $c$ it holds that either $V_c (H) \subseteq L$ or $V_c (H) \subseteq R$ (i.e., the vertices with $c$ in their lists are all in one side).
We next use the bipartite Moore bounds:

\begin{lemma}\label{lem:color-lists-upper-bound}
    Let $\hat{H}_S$ be a $p$-random blocked subgraph of $H$ with $\hat{m}$ edges.
    Then $\E[\hat{m}] = O((p^\nu n)^{1+1/k})$.
\end{lemma}
\begin{proof}
    The proof is identical to \Cref{lem:VCFT-upper-bound}, but with $L_S = \{v \in L \mid \ell(v) \subseteq S\}$, so $\E[|L_S|] = p^\nu |L|$, and similarly for $R_S$.
\end{proof}

We are now ready to conclude the induction step.
Let $\hat{H}_S$ be a $p$-random blocked subgraph with $\hat{m}$ edges, where $p = 1/(2f)$.
For $i\in [2\nu, 2\nu + \mu]$, let
\begin{align*}
    E_i &= \{ e = \{u,v\} \in E(H) \mid |\ell(e) \cup \ell(u) \cup \ell(v)| = i\}, \\
    B_i &= \{ (e,c) \in B \mid e \in E_i\}.
\end{align*}
Note that $E(H) = \bigcup_i E_i$ (as there are no conflicting edges), $B = \bigcup_i B_i$, and $|B_i| \leq f |E_i|$.
We have:
\[
\E[\hat{m}]
\geq \E \sparen{ |E(H[S])| - |B_S|}
= \sum_{i=2\nu}^{2\nu + \mu} \Big( \frac{|E_i|}{2^i f^i} - \frac{|B_i|}{2^{i+1} f^{i+1}} \big)
\geq \sum_{i=2\nu}^{2\nu + \mu} \frac{|E_i|}{2^{i+1} f^i}
= \Omega\Big( \frac{m}{f^{2\nu + \mu}} \Big).
\]
Combining with \Cref{lem:color-lists-upper-bound}, and rearranging yields $m = O(f^{\mu + \nu(1-1/k)} n^{1+1/k})$, which concludes the induction step and proves \Cref{thm:color-lists}.

\subsection{Lower Bound}

The lower bound for the color lists setting is obtained by combining the proofs of \Cref{thm:edge-color-spanner-lower-bound} and \Cref{thm:mixed-color-spanner-lower-bound}.

\begin{theorem}\label{thm:list-color-spanner-lower-bound}
    Let $\mu, \nu$ be constant integers.
    Assuming the Girth Conjecture, 
    there exists a \emph{simple} $(\mu,\nu)$-list colored graph $G$ with $n$ vertices,
    such that every $f$-CFT $(2k-1)$-spanner of $G$ must have $\Omega(\min\{f^{\mu + \nu(1-1/k)} n^{1+1/k}, n^2\})$ edges.
\end{theorem}

\begin{proof}[Proof sketch.]
    First, repeat the proof of \Cref{thm:edge-color-spanner-lower-bound} to create an $n$-vertex graph $G'$ consisting of $\binom{f+\mu}{\mu}$ edge-disjoint ``Erd\H{o}s graphs" with girth $\geq 2k+2$ and $\Omega(n^{1+1/k})$ edges.
    Let $\C_E$ be a set of $f+\mu$ colors.
    Each Erd\H{o}s graph in $G'$ is associated with a $\mu$-subset of $\C_E$, which is the list of colors given to all its edges.

    Next, we modify the proof of \Cref{thm:mixed-color-spanner-lower-bound} using $G'$ as the base graph.
    Again, we may assume $G'$ is bipartite with sides $V(G') = L \cup R$.
    The lower bound graph $G$ will use the color palette $\C_E \cup \C_L \cup \C_R$, where these sets are disjoint, and $|\C_L|=|\C_R|=f+\nu$.
    The vertex set is
    \[
    V(G) = \paren{L \times \binom{\C_L}{\nu}} \cup \paren{R \times \binom{\C_R}{\nu}}.
    \]
    For each $s_L \in \binom{\C_L}{\nu}, s_R \in \binom{\C_R}{\nu}$, we place a copy of $G'$ on $(L \times \{s_L\} \cup (R \times \{s_R\})$.
    All vertices in $L \times \{s_L\}$ get $s_L$ as their color list, and similarly for $R \times \{s_R\}$.
    
    Now, let $s_E \in \binom{\C_E}{\mu}$.
    Failing the $3f$ colors $(\C_E - s_E) \cup (\C_L - s_L) \cup (\C_R - s_R)$ leaves us with a single Erd\H{o}s graph with girth $\geq 2k+2$.
    Every edge of $G$ lies in such an Erd\H{o}s graph, which shows that the $(\mu,\nu)$-list-colored graph $G$ has no proper $3f$-CFT $(2k-1)$-spanner.
    Note that $\binom{f+\mu}{\mu} = \Theta(f^\mu)$, and $\binom{f+\nu}{\nu} = \Theta(f^\nu)$ (as $\mu,\nu$ are constants).
    So, $|V(G)| = \Theta(f^\nu n)$ vertices, and $|E(G)| = \Theta(f^{2\nu + \mu} n^{1+1/k}) = \Theta(f^{\mu+\nu(1-1/k)} |V(G)|^{1+1/k})$.
    The result follows by rescaling.
\end{proof}

\end{document}